\documentclass[a4paper,10pt]{llncs}

\usepackage[english]{babel}
\usepackage{url}
\usepackage[utf8]{inputenc}
\usepackage{subfigure}
\usepackage{latexsym}
\usepackage{amsmath}
\usepackage{amssymb}
\usepackage{graphicx}
\usepackage{verbatim}
\usepackage{graphicx}
\usepackage{llncsdoc}
\usepackage{wrapfig}

\begin{document}

\title{Single and multiple consecutive permutation motif search
}

\author{%
    Djamal Belazzougui\inst{1}
    \and
    Adeline Pierrot\inst{2}
    \and
    Mathieu Raffinot\inst{2}
    \and
    Stéphane Vialette\inst{3}
}


\institute{}

\institute{%
   Department of Computer Science, FI-00014 University of Helsinki, Finland
  \email{}
  \and
  LIAFA, Univ. Paris Diderot - Paris 7, 75205 Paris Cedex 13, France
  \and
  LIGM CNRS UMR 8049,
  Universit\'e Paris-Est, France \\
  \email{vialette@univ-mlv.fr}
}

\date{\today}
\maketitle


\begin{abstract}
Let $t$ be a permutation (that shall play the role of
the {\em text}) on $[n]$ and a pattern $p$ be a sequence of
$m$ distinct integer(s) of $[n]$, $m\leq n$. The pattern $p$ occurs in
$t$ in position $i$ if and only if $p_1 \ldots p_m$ is
order-isomorphic to $t_i \ldots t_{i+m-1}$, that is, for all $1 \leq
k< \ell \leq m$, $p_k>p_\ell$ if and only if $t_{i+k-1}>t_{i+\ell-1}$.
Searching for a pattern $p$ in a text $t$ consists in identifying all
occurrences of $p$ in $t$. We first present a forward automaton which
allows us to search for $p$ in $t$ in $O(m^2\log \log m +n)$
time. We then introduce a Morris-Pratt automaton representation of the
forward automaton which allows us to reduce this complexity to
$O(m\log \log m +n)$ at the price of an additional amortized
constant term by integer of the text. Both automata occupy $O(m)$
space. We then extend the problem to search for a set of patterns and
exhibit a specific Aho-Corasick like algorithm. Next we present a
sub-linear average case search algorithm running in
$O\left(\frac{m\log m}{\log\log m}+\frac{n\log m}{m\log\log m}\right)$
time, that we eventually prove to be optimal on average.
\end{abstract}


\section{Introduction}

Two sequences are \emph{order-isomorphic} if the permutations required
to sort them are the same.
A sequence $p$ is said to be a \emph{pattern} (or \emph{occurs})
within a sequence $t$ if $t$ has a \emph{subsequence} that is
order-isomorphic to $p$.
Pattern involvement \emph{permutations} and sequences has now become
a very active area of research \cite{Kitaev:2011}.
However, only few results on the complexity of finding patterns in
permutations and sequences are known.
It appears to be a difficult problem
to decide of two given permutations $\pi$ and $\sigma$
whether $\sigma$ occurs in $\pi$, and in this generality the problem is NP-complete
\cite{Bose:Buss:Lubiw:1998}.
For $\sigma \in S_m$ and $\pi \in S_n$,
the $O(n^m)$ time brute-force algorithm was improved to
$O(n^{0.47m + o(m)})$ time in
\cite{Ahal:Rabinovich:2008}.
There are several ways in which this notion of permutation patterns may be generalized,
and we focus here on \emph{consecutive patterns}
(\emph{i.e.} the match is required to consist of contiguous elements)
 \cite{Kitaev:2011}.
A sequence $p$ is said to be a \emph{consecutive pattern} or \emph{consecutively occurs}
within a sequence $t$ if $t$ has a \emph{substring} that is
order-isomorphic to $p$.
Searching for a pattern $p$ in a text $t$
consists in identifying all occurrences of $p$ in $t$.
Recently, using a modification of the classical Knuth-Morris-Pratt string
matching algorithm, a $O(n + m\log m)$ time algorithm has been proposed
for checking if a given sequence $t$ of length $n$ contains a substring which is
order-isomorphic to a given pattern $p$ of length $m$
 \cite{Kubica:Kulczynski:Radoszewski:Rytter:Walen:2013}.
The time complexity reduces to $O(n + m)$ time under the assumption that the
symbols of the pattern can be sorted in $O(m)$ time.




The set of all integers from $1$ to $n$ is written $[n]$. Let $t$ be a
permutation of length $n$ and $p$ be a sequence of $m \leq n$ distinct
integers in $[n]$.
First we present a forward automaton
which allows us to search for $p$ in $t$ in $O(m^2\log \log m +n)$
time.
Next, we introduce a Morris-Pratt automaton
representation~\cite{MP70} of the forward automaton which allows us to
reduce this complexity to $O(m\log \log m +n)$ at the price of an
additional amortized constant term by integer of the text. Both
automata occupy $O(m)$ space. We then extend the problem to search for
a set of patterns and exhibit a specific Aho-Corasick like
algorithm. Finally we present a sub-linear average case search algorithm
running in $O(n \log m / \log \log m)$ time that we eventually prove
to be optimal on average.

Let us define some notations. The set of all permutations on $[n]$ is
denoted by $S_n$.  Let $\Sigma_n = [n]$.  Abusing notations, we
consider in this paper permutations of $S_n$ as strings without symbol
repetition, and we denote by $\Sigma_n^*$ the set of all strings
without symbol repetition (including the empty string), where each
symbol is an integer in $[n]$.  A {\em prefix} (resp. {\em suffix},
{\em factor}) $u$ of $p$ is a string such that $p=u w, w \in
\Sigma_n^*$. (resp. $p= w u, w \in \Sigma_n^*$, $p= wuz, w,z \in
\Sigma_n^*$. We also denote $|w|$ the number of integer(s) in a string
$w, w \in \Sigma_n^*$.  We eventually denote $p^r$ the reverse of $p$,
that is, the string formed by the symbols of $p$ read in the reverse
order.
We denote by $p^{\equiv}$ the set of words of $\Sigma_n^*$ which are
order-isomorphic to $p$.

The following property is useful for
designing automaton transitions.

\begin{property}
\label{trans}
Let $p=p_1 \ldots p_m \in \Sigma_n^*$ and $w=w_1\ldots w_\ell \in \Sigma_n^*$,
$\ell< m$,
such that $w$ is order-isomorphic to $p_1 \ldots p_\ell$, and
let $\alpha \in \Sigma$. Testing if $w\alpha$ is order-isomorphic to
$p_1 \ldots p_\ell p_{\ell+1}$ can be performed in constant time storing only
a pair of integers.
\end{property}

\noindent
{\bf Proof.} The pair of integers $(x_1,x_2)$ is determined as follows:
$x_1 \leq \ell$ is the greatest number such that $p_{x_1}$ is the position
of one of the largest integer in $p_1..p_\ell$ which is smaller than
$p_{\ell+1}$, if any. Otherwise, we fix $x_1$ arbitrarily to
$-\infty$. Let $x_2 \leq \ell$ be the greatest position of one
of the smallest integer in $p_1..p_\ell$ which is larger than
$p_{\ell+1}$, if any. Otherwise, we fix $x_2$ to $+\infty$. Now, it
suffices to test if $ w_{x_1} < \alpha < w_{x_2}$ to verify if
$w\alpha$ is order-isomorphic to $p_1\ldots p_{\ell+1}$ \qed

We define a function $\mbox{rep}(p=p_1\ldots p_m,j)$ which returns a
pair of integers $(x_1,x_2)$ that represents the pair defined in
property \ref{trans} for the prefix of length $j$ of a motif $p$.

\section{Tools}
\label{sec:tools}

Before proceeding, we first describe some useful data structures we shall
use as basic subroutines of our algorithms. The problem called
\emph{predecessor search problem} is defined as follows: given a set
$S=\{x_1,x_2,\ldots x_n\}\subset [u]$ ($u$ is called the size of
the universe), we support the following query: given an integer $y$
return its predecessor in the set $S$, namely the only element $x_i$
such that $x_i\leq y\leq x_{i+1}$~\footnote{By convention, if all the
  elements of $S$ are smaller than $y$, then return $-\infty$ and if
  they are larger than $y$ then return $x_n$}. In addition, in the
dynamic case, we also support updates: add or remove an element from
the set $S$.
The standard data structures to solve the predecessor search are the
balanced binary search trees~\cite{AVL63,Ba72}.  They use linear space
and support queries and updates in worst-case $O(\log n)$
time. However, there exists better data structures that take advantage
of the structure of the integers to get better query and update
time. Specifically, the Van-Emde-Boas tree~\cite{Bo77} supports
queries and updates in (worst-case) time $O(\log\log u)$ using $O(u)$ space. Using
randomization, the y-fast trie achieves linear space with queries
supported in time $O(\log\log u)$ and updates supported in randomized
$O(\log\log u)$ time.  The problem has received series of improvements
which culminated with Andersson and Thorup's result~\cite{AT07}. They
achieve linear space with queries and updates supported in
$O(\mathrm{min}(\log\log u,\sqrt{\frac{\log n}{\log\log n}}))$ (the
update time is still randomized).

A special case occurs when space $n$ is available and the set of keys $S$
is known to be smaller than $\log^c n$ for some constant $c$. In this
case all operations are supported in worst-case constant time using
the atomic-heap~\cite{Wi20}.
\section{\vspace*{-0.2cm}Forward search automaton}

The problem we consider is to search for a motif $p$ in a permutation
$t$ without preprocessing the text itself. By analogy to the simpler
case of the direct search of a word $p$ in text $t$, we build an
automaton that recognizes $(\Sigma_n^*)\dot p^{\equiv}$. We then prove its size
to be linear in the length of the pattern.

We formally define our forward search automaton ${\mathcal FD}(p)$ built
on $p=p_1 \ldots p_m$ as follows:
\begin{itemize}
\item $m+1$ states corresponding to each prefix (including the empty
  prefix) of $p$, state $0$ is initial, state $m$ is terminal;
\item $m$ forward transitions from state $j$ to $j+1$ labelled by
  $\mbox{rep}(p,j+1)$;
\item $bt$ backward transitions $\delta(x,[i,j])$, where $x$ numbers
  a state, $ 0 \leq x \leq m$, $i \in {1,\ldots,x} \cup {-\infty}$,
  $j \in {1,\ldots,x} \cup {+\infty}$, defined the following way:
  $\delta(x,[i,j]) = q$ if and only if for all $p_i < \alpha < p_j$
  (resp. $k=\alpha < p_j$ if $i = -\infty$, $p_i < \alpha$ if $j =
  +\infty$), the longest prefix of $p$ that is order-isomorphic to a
  suffix of $p_1 \ldots p_x \alpha$ is $p_1 \ldots p_q$.
\end{itemize}
We also impose some constraints on outgoing transitions. Let $x$ be a
given state corresponding to the prefix $p_1 \ldots p_x$.

 Let us sort all
$p_i, 1 \leq i \leq x$ and consider the resulting order
$p_{i_0}=-\infty<p_{i_1}<\ldots <p_{i_k}<+\infty=p_{i_{k+1}}$. We
build one outgoing transition for each interval
$[p_{i_j},p_{i_{j+1}}]$, excepted if $p_{i_{j+1}}=p_{i_{j}}+1.$ Also we
merge transitions from the same state to the same state that are labeled by
consecutive intervals.

It is obvious that the resulting automaton recognizes a given pattern
in a permutation by reading one by one each integer and choose the
appropriate transition. Figure \ref{aut1} shows such an automaton.

\begin{figure}[h]
  \centering
\includegraphics[width=9cm]{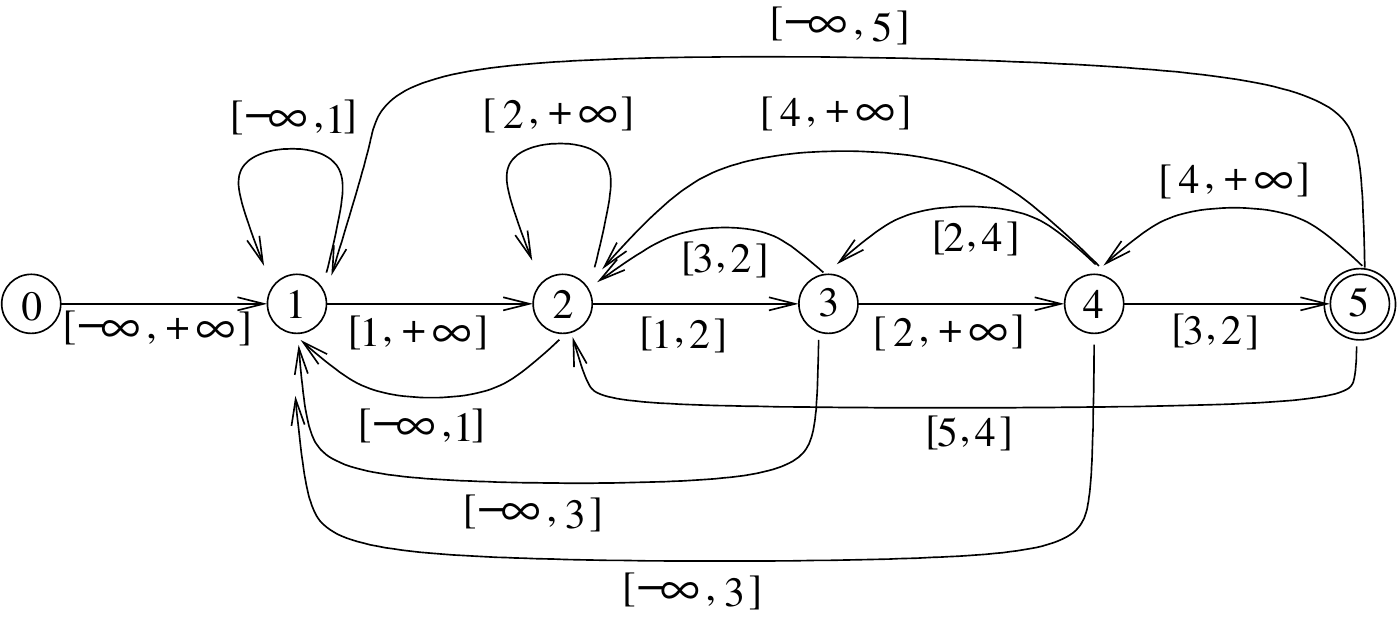}
\caption{Forward automaton built on $p=4 \,12 \,6 \,16 \,10$. State $0$ is
  initial and state $5$ is terminal.}
 \label{aut1}
\end{figure}

%
%
%
%
The main result on the structure of the forward automaton is the following.

\begin{lemma}
\label{lemmasize}
The number of transitions of the forward automaton built on $p_1 \ldots
p_m$ is linear in $m$.
\end{lemma}

Lemma~\ref{lemmasize} combined with the fact that the outgoing
transitions from each state $q$ are sorted accordingly to the closest
proximity to $q$ of their arrival state leads to the following lemma.

\begin{lemma}
\label{searchphase}
Searching for a consecutive motif $p=p_1 \ldots p_m$ in a permutation
$t=t_1 \ldots t_n$ using a forward automaton built on $p$ takes
$O(n)$ time.
\end{lemma}

We can build the forward automation in  $O(m^2\log \log m)$
time. However, we defer the proof of this construction for the following
reason. This $O(m^2\log \log m )$ complexity might be too large for
long patterns. Nevertheless, we show below that we can compute in a
first step a type of Morris-Pratt coding of this automaton which can
either (a) be directly used for the search for the pattern in the text
and will preserve the linear time complexity at the cost of an
amortized constant term by text symbol, or (b) be developed to build
the whole forward automaton structure.


Therefore we present and build a new automaton ${\mathcal MP}$ that is a
Morris-Pratt representation of the forward automaton. The idea is to
avoid building all backward transitions by only considering a special
backward single transition from each state $x, x>0$ named {\em
  failure} transition. We formally define our automaton ${\mathcal MP}(p)$
built on $p=p_1 \ldots p_m$ the following way:

\begin{itemize}
\item $m+1$ states corresponding to each prefix (including the empty
  prefix) of $p$, state $0$ is initial, state $m$ is terminal;
\item $m$ forward transitions from state $j$ to $j+1$ labelled by
  $\mbox{rep}(p,j+1)$;
\item $m$ failure transitions (non labelled) defined by:
  a failure transition connects a state $j>0$ to a state $k<j$ if and
  only if $p_1\ldots p_k$ is the largest order-isomorphic border of
  $p_1 \ldots p_j$.
\end{itemize}

\begin{figure}[h]
  \centering
\includegraphics[width=8cm]{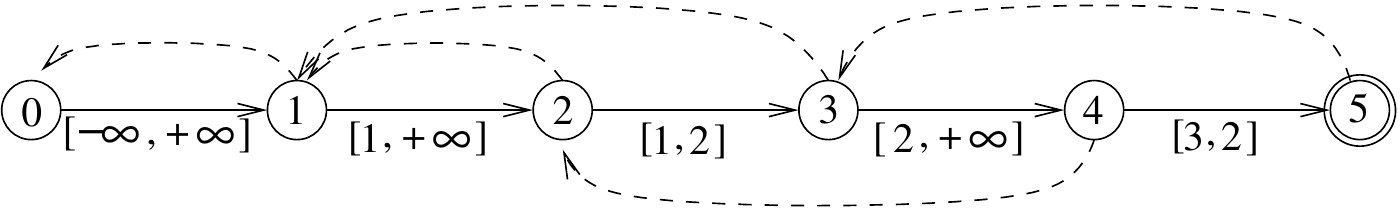}
\caption{${\mathcal MP}$ automaton built on $p=(4,12,6,16,10)$.
State $0$ is initial and state $5$ is terminal. Backward transitions are failure transitions.}
 \label{mp1}
\end{figure}

Reading a text $t$ through the MP representation of the forward
automaton is performed the following way. Let us assume we reached
state $x<m$ and we read a symbol $t_i$ at position $i$ of the text.
Let $[k,\ell] =\mbox{rep}(p,x+1).$ If $t_i \in [t_{i-m+k},t_{i-m+l}]$ we
follow the forward transition and the new current state is
$x+1$. Otherwise, we {\em fail} reading $t_i$ from $x$ and we retry
from state $q = \mbox{fail}(x)$ and so-on until (a) either $q$ is
undefined, in which case we start again from state $0$, either (b) a
forward transition from $q$ to $q+1$ works, in which case the next
current state is $q+1$.

\begin{lemma}
\label{kmp}
Searching for a pattern $p$ in a text $t_1 \ldots t_m$ using the
Morris-Pratt representation of the forward automaton built on $p$ is
$O(n)$ time.
\end{lemma}
In order to prove lemma \ref{kmp} we need
to focus on the classical notion of border that we extend to our framework.
\vspace*{-0.1cm}
\begin{definition}
 Let $p \in \Sigma_n^*$. A {\em border} of $p$ is a word $w\Sigma_n^*,
 |w|<|p|$ that is order-isomorphic to a suffix of $p$ but also
 order-isomorphic to a prefix of $p$.
\end{definition}

The construction of the forward automation relies of the maximal
border of each prefix that is followed by an appropriate integer in
the pattern. The Morris-Pratt approach is based on the following property:

\begin{property}
A border of a border is a border.
\end{property}

This property allows us to replace the direct transition of the
forward algorithm by a search along the borders, from the longest to
the smallest, to identify the longest one that is followed by the
appropriate integer.
We prove now that we can build the Morris-Pratt representation of the
forward automaton efficiently.

\begin{lemma}
\label{kmpbuild}
Building an Morris-Pratt representation of the forward automaton on a
consecutive motif $p=p_1 \ldots p_m$ can be performed in (worst-case)
$O(m \log\log m)$ time.
\end{lemma}

Lemma \ref{kmp} and \ref{kmpbuild} allow us to state
the main theorem of this section.

\begin{theorem}
\label{theoall}
Searching for a consecutive motif $p=p_1 \ldots p_m$ in a permutation
$t=t_1 \ldots t_n$ can be done in $O(m\log\log m+ n)$ time.
\end{theorem}

The Morris-Pratt representation of the forward automaton permits to
search directly in the text at the price of larger amortized
complexity (considering the constant hidden by the $O$ notation) than
that required by searching with the forward automaton directly. If the
real time cost of the search phase is an issue, the forward automaton
can be built form its Morris-Pratt representation as follows.

\begin{property}
\label{forward}
Building the forward automaton of a consecutive motif $p=p_1 \ldots
p_m$ can be performed in $O(m^2\log\log m)$ time.
\end{property}

An interesting point is that the construction of the forward automaton
from its Morris-Pratt representation can also be performed in a lazy
way, that is, when reading the text. The missing transitions are then
built {\em on the fly} when needed.

\section{Multiple worst case linear motif searching}

We can extend the previous problem defined for a single pattern to a
set of patterns $S$.
We note by $d$ the number of patterns, by $m$ the
total length of the patterns and by $r$ the length of the longest
pattern.
For this problem we adapt the Aho-Corasick
automaton~\cite{AC75} (or ${\mathcal AC}$ automaton for short).  The ${\mathcal AC}$ automaton
is a generalization of the ${\mathcal MP}$ automaton to a set of multiple
patterns. We note by $P$ the set of prefixes of strings in $S$.
In order to simplify the description we will assume that the set
of patterns $S$ is prefix-free. That is, we will assume that no pattern is prefix
of another. Extending the algorithm to the case where $S$ is non-prefix free, should not
pose any particular issue.
The states of the ${\mathcal AC}$ automaton are defined in the same way as in the $\mathcal MP$ automaton. Each state $t$ in the ${\mathcal AC}$ automaton corresponds uniquely to
a string $p\in P$. The forward transitions are defined as follows:
there exists a forward transition connecting state $s$ to each state
corresponding to an element $pc\in P$ (where $c$ is a single symbol). Thus
this definition of the forward transitions matches essentially the
definition of the forward transitions in the ${\mathcal MP}$ automaton.
The failure transitions are defined as follows: a failure transition
a state $s$ corresponding a string $p$ to the state
$s'$ corresponding to the longest string $q$ such that $q\in P$ and
$q\neq p$.  The matching using the ${\mathcal AC}$ automaton is done in the same
way as in the $\mathcal MP$ automaton using the forward and failure transitions.

\subsection{Our extension of the ${\mathcal AC}$ automaton}

We could use exactly the same algorithm as the one used previously for
our variant of the $\mathcal MP$ automaton with few differences.  We
describe our modification to ${\mathcal AC}$ automaton to adapt it to
the case of consecutive permutation matching. An important observation
is that we could have two or more elements of $P$ that are both of the
same length and order-isomorphic. Those two elements should have a
single corresponding state in the ${\mathcal AC}$ automaton. Thus, if
two or more elements of $P$ are order-isomorphic then we keep only one
of them. For the forward transitions, we can a associate a pair of positions
$(x_1,x_2)$ to each forward transition. Then we can check which
transition is the right one by checking the condition
$t_{i-m+x_1}<t_i<t_{i-m+x_2}$ for every pair $(x_1,x_2)$ and take the
corresponding transition. The main problem with this approach is that
the time taken would grow to $O(d)$ time to determine which transition
to take which can lead to a large complexity if $d$ is very large.
Our approach will instead be based on using a binary search tree (or
more sophisticated predecessor data structure). With the use of a
binary search tree, we can achieve $O(\log m)$ time to decide which
transition to take. More precisely, each time we read $T[i]$ we insert
the pair $(t_i,i)$ into the binary search tree. The insertion uses the
number $t_i$ as a key. Now suppose that we only pass through forward
transitions. Then a transition at step $i$ is uniquely determined by:
(1) the current state $s$ corresponding to an element $p\in P$; (2)
the position of the predecessor of $t_i$ among $t_{i-|p|}\ldots
t_{i-1}$.

To determine the predecessor of $t_i$
among $t_{i-|p|}\ldots t_{i-1}$,  the binary search tree should contain
precisely the $|p|$ pairs corresponding to $t_{i-|p|]}\ldots t_{i-1}$. If the predecessor
of $t_i$ in the binary search tree is a pair $(t_j,j)$, we
then conclude that the element $p[|p|-j+1]$ is the predecessor of $t_{i}$ in
$p$.

In order to maintain the binary search tree we must do the following
actions during passing through a failure or a forward transition:
(1) whenever we pass through a forward transition at a step $i$
we insert the pair $(t_i,i)$;
(2) whenever we pass through a failure transition from a state
corresponding to a prefix $p_1$ to a state corresponding to a prefix $p_2$,
then we should remove from the binary
tree all the pairs corresponding to the symbols $t_{i-|p_1|}\ldots t_{i-|p_2|}$.


It should be noted that each removal or insertion of a pair into
the binary search tree takes $O(\log r)$ time. The upper bound
$O(\log r)$ comes from the fact that we never insert more than $r$
elements in the binary search tree.  Since in overall we are doing $O(n)$
insertions or removals, the amortized time should simplify to
$O(n\log r)$.  Finally if we replace binary search tree with a more efficient
predecessor data structure, we will be able to achieve randomized time
$O(n\cdot t)$ where $t=\mathrm{min}(\log\log n,\sqrt{\frac{\log
    r}{\log\log r}},d)$ is the time needed to do an operation on the
predecessor data structure (see section~\ref{sec:tools} for details).
We use the linear space version of the predecessor data structure which guarantees
only randomized performance but uses $O(r)$ additional space only.
We thus have the following theorem :
\begin{theorem}
\label{theoACsearch}
Searching for set of $d$ consecutive motifs of maximal length $r$
and whose ${\mathcal AC}$ automaton has been built and where
the longest pattern is of length $r$ can be done in
randomized $O(nt)$ time, where $t=\mathrm{min}(\log\log n,\sqrt{\frac{\log
    r}{\log\log r}},d)$.
\end{theorem}

{\bf Preprocessing.}
We now show that the preprocessing phase can be done in worst-case
$O(m\log\log r)$ time.
As before our starting point will be to sort all the patterns and reduce
the range of symbols of each pattern of length $\ell$ from range $[n]$
to the range $[1..\ell]$. This takes worst-case time $O(m\log\log r)$.

Recall that two or more elements of $P$ of the same length and order-isomorphic
should be associated with the same state in the ${\mathcal AC}$ automaton.
In order to identify the order-isomorphic elements of $P$, we will carry a first step
called normalization. It consists in normalizing each pattern.
A pattern $p$ is normalized by
replacing each symbol $p_j$ by the pair $\mbox{rep}(p=p_1\ldots
p_{j-1},j)$ (consisting in the positions of the predecessor and
successor among symbols $p_1\ldots p_{j-1}$). This can be done for
all patterns in total $O(m\log\log r)$ time. In the next step, we build a
trie on the set of normalized patterns. This takes linear time. The
trie naturally determines the forward transitions. More precisely
any node in the trie will represent a state of the automaton
and the the labeled trie transitions will represent follow transitions.

Note that unlike the forward automaton (or the ${\mathcal MP}$ automaton)
there could be more than one outgoing forward transition from each node.
In order to encode the outgoing transition from each node,
we will make use of a hash table that stores all
the transitions outgoing from that node. More precisely
for each transition labeled by the pair $\mbox{rep}(p=p_1\ldots p_{j-1},j)$
and directed to a state $q$, the hash table will associate the key $p_1$
associated with the value $q$.
Now that the next transitions have been successfully built, the final
step will be to build the failure transitions and this takes more effort.
In order to build the failure transitions we decompose the
trie into $r$ layers. The first layer consists in the nodes of the trie that
represent prefixes of length $1$. The second layer consist in all the nodes
that represent prefixes of length $2$, etc.

Next, we will reuse the same algorithm that was used in~\ref{kmpbuild} to
build the ${\mathcal MP}$ automaton but adapted to work on
the ${\mathcal AC}$ automaton.
Instead of using a single predecessor data structure we will use multiple
predecessor data structures and attach a pointer to a predecessor data structure
at each trie node.
A node of the original non compacted trie will share the same predecessor
with its parent, iff it is the only child of its parent.
The following building phases will no longer reuse the normalized patterns, but
instead reuse the original patterns. To each node, we attach a pointer
to one of the original pattern.
More precisely if a node has a single child, then his pattern pointer
will be the same as its (only) child pattern pointer. If a node has more than one child
(in which case it is called a \emph{branching node}), then it will point to the
shortest pattern in its subtree. If a node is a leaf then it will directly point
to the corresponding pattern.
A predecessor data structure of a node whose pattern pointer points to a pattern
of length $u$ will have capacity to hold $u$ keys from universe $u$ and thus will
use $O(u)$ space. This is justified by the fact that the predecessor data structure
will only hold at most $u$ elements of the patterns and each element value
is at most $u$ (recall that the pattern is a permutation of length $u$).

In order to bound the total number of predecessor data structures and their total size,
we consider a compacted version of the trie (Patricia trie), where
each node with a single child is merged with that single child.
A node in the original (non-compacted) trie with two of more children is called
\emph{branching node}. It is clear that the set of nodes of a patricia (compacted)
trie are precisely the branching nodes and the leaves of the original trie.

It is a well known fact that a Patricia trie with $r$ leaves has at most $2r-1$
nodes in total. Thus the total number of predecessor data structures will be upper
bounded by $2r-1$.
During the building if a node at layer $t$ has a single child, then
that single child at level $t+1$ will inherit the predecessor data structure
of its parent. Otherwise if the node $v$ at level $t$ has two or more children
at level $t+1$, then a predecessor data structure is created for each child $u$.
Then if the predecessor data structure of $v$ contains exactly $k$ elements,
those elements are precisely $x_{t+1-k}\ldots x_t$, where $x$ is string pointed
by $v$. We will insert the k elements $y_{t+1-k}\ldots y_t$
into the predecessor data structure of $u$, where $y$ the string pointed by $u$.


In order to bound the total space used by the predecessor data
structures, we notice that the total capacities of all predecessor
data structures is $O(m)$. This can easily be proved.
Because we know that the total length of all patterns
is bounded by $m$, we will also know that the total cumulative length of all strings pointed
by branching node is also upper bounded by $m$. This is because precisely the pointed
strings are precisely the shortest strings in the subtrees rooted by the branching node.
The same holds for the leaves as the capacities of their respective predecessor data structures
will be no more than the total length of the patterns that correspond to the leaves which is $O(m)$.

We finally need to bound the total construction time which is dominated by the operations on
the predecessor data structures. The time is clearly bounded by $O(m\log\log r)$.
This is by a straightforward argument: as the total sum of the pointed strings is $O(m)$,
and we know that each element of a pointed string can only be inserted or deleted once,
and furthermore each insert/delete cost precisely $O(\log\log r)$ worst-case time, we conclude that
the total time spent in the predecessor data structure is worst case $O(m\log\log r)$.
We thus have the following theorem:

\begin{theorem}
\label{theACbuild}
Building the ${\mathcal AC}$ automaton for a set of $d$ consecutive
motifs of total length $m$ and where the longest motif is of length
$r$ can be done in worst-case $O(m\log\log r)$ time.
\end{theorem}

\section{Single sublinear average-case motif searching}

Algorithm forward takes $O(n+m\log \log m)$ time in the worst case time but also on
average. We present now a very simple and efficient average
case-algorithm which takes
$O(\frac{m\log m}{\log\log m}+n\frac{\log m }{m \log \log m})$ time.

In order to search for a pattern $p$ in $t$, we first build a tree
$T$ of all isomorphic-order factors of $p^r$ of length $\frac{3.5\log m}{\log
\log m}$. $T$ is built by inserting each such factor one after the
other in a tree and building the corresponding path if it does not
already exist. The construction of this tree requires $O(\frac{m\log m}{\log\log m})$
time (details are given below). The search phase is performed through a window of size $m$ that is
shifted along the text. For each position of this window,
$b = \frac{3.5\log m}{\log \log m}$ symbols are read backward from the end
of the window in the tree $T$. Two cases may occurs: {\em (i)} either the factor is not recognized as a factor of $p^r$. This
  means that no occurrence of $p$ might overlap this factor and we can
  surely shift the search window after the last symbol of this
  factor; {\em (ii)} either the factor is recognized, in which case we simply check
  if the motif is present using a naive $O(m)$ algorithm, and we
  repeat this test for the next $O(m/2)$ symbols. This might
  require $O(m^2/2)$ steps in the worst case.
  Figure \ref{fail} illustrates the first case.

\begin{figure}[h]
  \centering
\includegraphics[width=9cm]{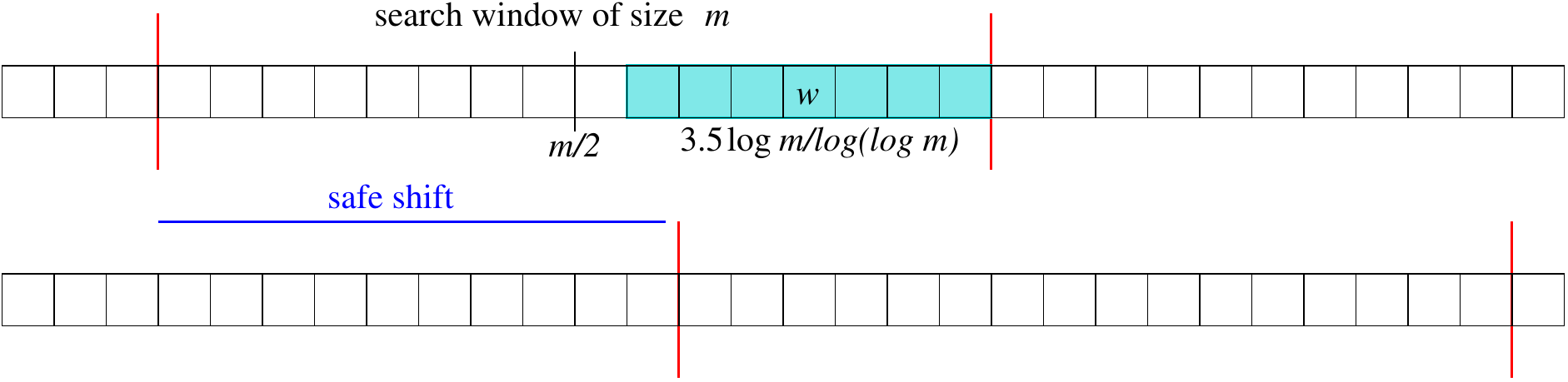}
\caption{First case: the motif $w^r$ is not recognized in the tree
  $T$, which implies that no occurrence of $p$ can overlap $w$ and the search
  window can surely be shifted after the first symbol of $w.$}
 \label{fail}
\end{figure}

Let us analyze the average complexity of our algorithm, in a
Bernoulli model with equiprobability of letters, that is, every
position in the text and the paper is independent of the others and
the probability of a symbol to appear is $1/\sigma.$ We also
consider that $ b < m/2$ since we are interested in analyzing the
average complexity for pattern long enough.

We count the average number of symbol comparisons required to shift
the search window of $m/2$ symbols to the right. As there are
$2n/m$ such segments of length $m/2$ symbols in $n$, we will simply
multiply the resulting complexity by $2n/m$ to gain the whole average
complexity of our algorithm.

There might be $O(b!)$ distinct motifs that could appear in the text while this
number is bounded by $m-b+1$ in the pattern (one by position). Thus,
with a probability bounded by $\frac{m-b +1}{b!}$ we will recognize
the segment of the text as a factor of $p$ and enter case 2. In which
case, moving the search window of $m/2$ symbols to the right using
the naive algorithm will require $O(m^2/2)$ worst case time.

In the other case which occurs with probability at least
$1-\frac{m-b+1}{b!}$, shifting the search window by $m/2$
symbols to the right only requires reading $b$ numbers.

The average complexity  (in terms of number of symbol reading and comparisons)
for shifting by $m/2$ symbols is thus (upper) bounded by $$A = O((m^2/2) \frac{m-b +1}{b!}+ b (1-\frac{m-b+1}{b!}))$$
and the whole complexity by $O((2n/m)A)$. By expanding and simplifying $A$
we get that $A = O(b + O(m^3/2b!)).$ Now using the famous Stirling
approximation $\ln(m!)=m\ln m-m+O(\ln m)$, it is not difficult to
prove that $b! = 2^{b\log b-b\log e+O(\log b)}= \Omega(m^3)$  and thus $A = O(b)$ and the whole
average time complexity (in terms of number of symbol reading and comparisons)
turns out to be $O(\frac{n\log m}{m \log \log m})$.

\subsection{Implementation details}

The tree $T$ can actually be built in $O(\frac{m\log m}{\log\log m})$ time by using
appropriate data structures.  Recall that the tree $T$ recognizes all the factors
of $p^r$ of length $\frac{3.5\log m}{\log
\log m}$. To implement $T$, we use the
same ${\mathcal AC}$ automaton presented in previous section to build the tree $T$,
but with two differences: we only need forward transitions and the length of any
pattern is bounded by $\frac{\log m}{\log\log m}$. Thus the cost is upper
bounded by $O(\frac{m\log m}{\log\log m}\cdot t)$, where $t$ is the time needed to do an operation
on the predecessor data structure (maximum of the times needed for inserts/deletes
and searches)
We now turn our attention to the cost of the matching phase.
From the previous section, we know that the total complexity
in terms of number of symbol reading and comparisons
is $O(\frac{n\log m}{m\log\log m})$.
The total cost of the matching phase is dominated by the multiplication of the total number
of text symbols read multiplied by the cost of a transition in the ${\mathcal AC}$
automaton which itself is dominated by the time to do an operation
on a predecessor data structure. The total cost of the matching phase is
thus $O(\frac{n\log m}{m\log\log m}\cdot t)$, where $t$ is the time needed to do an operation
on the predecessor data structure.

Now the performance of both matching and building phases crucially depend on
the used predecessor data structure.
 If a binary search tree is used then $t=O\left(\log\frac{\log m}{\log\log m}\right)=O(\log\log m)$
and the total matching time becomes $O(nt)=O(n\log\log m)$, and the total
building time becomes $O(m\log m)$.
However, we can do better if we work in the word-RAM model. Namely, we can use the
atomic-heap (see section~\ref{sec:tools}) which would add additional
$o(m)$ words of space and support all operations (queries, inserts and
deletes) in constant time on sets of size $\log^{O(1)} m$.
In our case, we have a set of size $O(\frac{\log m}{\log\log
  m})$ and thus the operations can be supported in constant time.
We thus have the following theorem:
\begin{theorem}
\label{theoall2}
Searching for a consecutive motif $p=p_1 \ldots p_m$ in a permutation
$t=t_1 \ldots t_n$ can be done in average $O(\frac{m\log m}{\log\log
  m}+\frac{n\log m}{m\log\log m})$ time.
\end{theorem}

\section{Average optimality}

We prove in this section a lower bound on the average complexity of
any consecutive motif matching algorithm. The proof of this bound is
inspired by that of Yao \cite{Yao79} which proved on average lower bound of
$O(\frac{n\log_{|\Sigma|}m}{m})$ for matching a pattern of length $m$
in a text of length $n$, both taken on a alphabet $\Sigma$. We
prove in our case of interest an average lower bound of $O(\frac{n\log
  m}{m \log \log m})$ considering all permutations over $[1\ldots n]$
to be equiprobable. As this average complexity is reached by the
algorithm we designed in the previous section, this bound is tight.

We begin to circumscribe our problem on small segments of length
$2m-1$ of the text into which we search for. Precisely, following
\cite{Yao79,NF04}, we divide our text in $\lfloor n/(2m-1)\rfloor$
contiguous and no-overlapping segments $s_i, 1 \leq \leq \lfloor
n/(2m-1)\rfloor$, such that $s_i(t)=t_{(2m-1)(i-1)+1} \ldots
t_{(2m-1)i}.$
When searching for a pattern in $t$, there might be occurrences
overlapping two blocks. But as we are interested on a lower bound, the
following lemma allows us to focus on all segments.

\begin{lemma}
\label{pre-counting lemma-1}
A lower bound for finding a pattern $p$ inside all segments $s_i(t)$ is
also a lower bound to the problem of searching for all occurrences of
$p$ in $t$.
\end{lemma}

We now prove that instead of focusing on all segments $s_i(t)$, we can
focus on obtaining a lower bound to search $p$ in any single segment and
then extend the lover bound on searching for $p$ inside this segment to
searching for $p$ inside all segments, and thus, using the previous lemma,
to the whole text.

\begin{lemma}
\label{pre-counting lemma-2}
The average time for searching for $p$ inside all segments $s_i(t)$ is
$\lfloor n/(2m-1) \rfloor$ times the average time for searching for $p$
inside any such segment.
\end{lemma}

Let $E(t)$ be the average complexity for searching $p$ in any
segments. Using the previous lemma, the whole average complexity is
$\sum_{i=1}^{\lfloor n/(2m-1) \rfloor} E(t) = \lfloor n/(2m-1)
\rfloor E(t) = \Omega (n/m) E(t).$

We now prove a lower bound for $E(t)$, which, using the two previous
lemma, gives us a lower bound for the whole problem.
Let ${\cal P}_m(\ell)$ the number of permutations of size $m$ that can be discarded
using a sliding window of size $m$ over a text of size $2m-1$ and checking only  $ 0< \ell \leq m$
positions in this window.




\begin{lemma}[Counting lemma]
\label{counting lemma}
Let $0<\ell \leq m$. Then
$$ |{\cal P}_m(\ell)| \; \leq \; m! \left( 1-\frac{1}{\ell!}
\right)^{\left\lceil{ \frac{m-1}{\ell^2}} \right\rceil}.$$
\end{lemma}

Let us consider now the whole set $S_m$ of permutations of
length $m$ which contains $m!$ such permutations. Given $1<l(m)\leq m$,
this set is the union of two distinct set ${\cal P}_m(\ell)$ and
$S_m \setminus {\cal P}_m(\ell),$ that is the set of motifs
discarded by a certificate of length l (or by $l$ accesses) and the
others. For all pattern in ${\cal P}_m(\ell),$ the average complexity
to be discarded is counted $1$. For any other motif in $S_m \setminus
{\cal P}_m(\ell),$ the average complexity is at least $l+1$.

The average complexity for discarding all patterns in $S_m$ is
thus $C(m)=\frac{|{\cal P}_m(\ell)|+ (m!- |{\cal P}_m(\ell)|)(l+1)}{m!}.$
We aim to find $l(m)$ that maximizes this expression when $m$ grows,
which will provide us a lower bound for the whole average complexity.
Now let us consider a fixed $l(m)$. We need to lower bound $C(m)$. As
$C(m)$ decreases when ${\cal P}_m(\ell)$ increases, this lower bound
is minimal when ${\cal P}_m(\ell)$ is as large as possible. Then, as
the counting lemma states that $|{\cal P}_m(\ell)| \; \leq \; m! \left(
1-\frac{1}{\ell!} \right)^{\left\lceil{ \frac{m-1}{\ell^2}}\right\rceil},$
$C(m)$ is minimal when
 $|{\cal P}_m(\ell)| \; = \; m! \left(
1-\frac{1}{\ell!} \right)^{\left\lceil{ \frac{m-1}{\ell^2}}\right\rceil}.$
We now arbitrarily impose $ 98/100 \leq \frac{|{\cal P}_m(\ell)|}{m!}\leq 99/100$. With
the left constraint, $C(m)\geq l - 98/100l +1 = \Omega(l)$. We want to compute $l(m)$ such that
$$98/100  \; \leq \;  \frac{|{\cal P}_m(\ell)|}{m!} = \left( 1-\frac{1}{\ell!} \right)^{\left\lceil{ \frac{m-1}{\ell^2}} \right\rceil} \leq 99/100.$$
Let us impose $\left\lceil \frac{m-1}{\ell^2}
\right\rceil \times \frac{1}{l!} \leq 1/10~ (ineq. 1)$.
%
This allows us to approximate our equation using the classical formula
$(1+x)^a=1+ax+\frac{a(a-1)}{2!}x^2 + \ldots + \frac{a!}{n!(a-n)!}x^n = 1+ax+\gamma$ where $a = \left\lceil
  \frac{m-1}{\ell^2} \right\rceil$, $x = \frac{-1}{l!}$ and $\gamma=\sum_{i=2}^n\frac{a!}{i!(a-i)!}x^i$.
It is easy to see that inequality (1) implies that $\gamma$ converges and is dominated by its first term which is bounded $\frac{a(a-1)}{2!}x^2\leq 1/200$.
We thus deduce that $(1+x)^a\in[1+ax,1+ax+1/200]$ which implies that $(1+x)^a-1/200\leq 1+ax\leq (1+x)^a$.
From
$(1+x)^a=\frac{|{\cal P}_m(\ell)|}{m!}\in[\frac{98}{100},\frac{99}{100}]$,
we obtain
$\frac{98}{100}-\frac{1}{200}\leq 1+ax\leq \frac{99}{100}$.
%
%
By replacing $a$ and $x$ in $1+ax$ we get :
$$\frac{98}{100}-\frac{1}{200}=195/200\; \leq \; 1-\left\lceil \frac{m-1}{\ell^2}\right\rceil \times \frac{1}{l!}\leq 99/100.$$
%
%
%
%
%
%
We prove in appendix that $l= \frac{b \log m}{\log \log m}$ with $b =
1 + o(1)$ verify these two inequalities and inequality (1). Thus
$O(\frac{n\log m}{m\log \log m})$ is a lower bound of the whole
average complexity for searching for a consecutive motif in a
permutation.










\subsection*{Acknowledgements} We would like to thanks Carine Pivoteau, Cyril Nicaud and Elie de Panafieu for checking parts of our calculus.


\bibliographystyle{plain}


\newpage
\section*{Appendix}

\bigskip
\begin{center}
   \rule{10cm}{1pt}.
\end{center}
\bigskip
\begin{proof}[Of Lemma~\ref{lemmasize}]

\noindent
{\bf Point 1.} We adapt the technique of \cite{Sim94} to our
framework. Let $q =\delta(x,[i,j])$ a backward transition from $x$ to
$q$ such that $q\geq 2$. Then $p_1 \ldots p_{q-1}$ is order-isomorphic
to the suffix of $p_1 \ldots p_x$ of length $q-1$. But either (a) $p_1
\ldots p_{q}$ is not order-isomorphic with $p_1 \ldots p_x$, or (b)
$x=m$ ($x$ is the last state of the automaton. Let $\ell=x-q.$ We prove
now {\em a contrario} that no other backward transition $q'
=\delta(x',[i',j'])$ such that $q'\geq 2$ can accept the same
difference $\ell'=x'-q'=\ell$. Let $q' =\delta(x',[i',j'])$ be such a
transition and consider without lost of generality that $2 \leq q'<q.$
Then $p_1 \ldots p_{q'-1}$ would be order-isomorphic to the suffix of
$p_1 \ldots p_{x'}$ of length $q-'1$, and $p_1 \ldots p_{q'}$ must not
be order-isomorphic to $p_1 \ldots p_{x'}p_{x'+1}.$ However, as $2 \leq
q'<q,$ $p_1 \ldots p_{q'}$ is a prefix of $p_1 \ldots p_{q-1}, $ and
as $l'=l',$ $p_1 \ldots p_{q-1}$ is order-isomorphic to the prefix of
$p_1 \ldots p_x$ of length $q'$, which is exactly $p_1 \ldots
p_{x'}p_{x'+1}.$ This leads to a contradiction and for a given $1\leq
\ell<m,$ there exists at most one backward transition $q
=\delta(x,[i,j]), q\geq 2$ such that $x-q=\ell$. This bounds the number
of such backward transition to $m-2$. Let $N(x)$ be the number of
backward transitions $q =\delta(x,[i,j])$ from $x$ such that $q\geq
2$.

\medskip
\noindent
{\bf Point 2.} We consider now all backward transitions $1 =\delta(x,[i,j])$
reaching state $1$. We denote such a transition a 1-transition. Note that state $0$ is
never reached by any transition because any two integers are always
order-isomorphic. The key observation is that from each state $x$
source of the transition, the number of such 1-transitions from $x$ is bounded by
$N(x)+2.$ This is true since 1-transitions and other transitions must be
interleaved to cover $[-\infty,+\infty].$ Therefore, as the total number of
$N(x)$ is bounded by $m-2$, the number of 1-transitions is bounded by $2m-4.$

\medskip
\noindent
{\bf Point 3.} The number of forward transitions is $m+1$, thus
the whole number of transitions is bounded by $4m-5.$
\qed
\end{proof}

\bigskip
\begin{center}
   \rule{10cm}{1pt}.
\end{center}
\bigskip

\begin{proof}[Of Lemma~\ref{searchphase}]
Searching for $p$ in $t$ using the forward automaton of $p$ can be
easily done reading all symbols of the text one after the other. But
at each state one must identify the right outgoing transition, which
normally requires to search in a list or an AVL tree. This would add a
polylog factor to all integer reading and thus the complexity would be
of the form $O(n. \mbox{polylog}(m)).$ However, the structure of the
forward automaton combined with the fact that we imposed all outgoing
transitions of each node to be sorted increasingly to the length of
the transition allow us to amortize the search complexity of the
searching phase along the permutation. The resulting search phase
complexity is $O(n)$ time. Indeed, let us search $t$ through the
automaton, reading one symbol at a time reaching a current state
$x$. Let us assume we read the text until position $i$ and we want to
match $t_{i+1}$. We test if $t_{i+1}$ belongs to the interval $[i,j]$
labeling $x+1=\delta(x,[i,j])$ if $x<m$. If yes, we follow this
forward transition. If not, we test each backward transition from $x$
in increasing length order.

\begin{figure}[htb]
  \centering
\includegraphics[width=6cm]{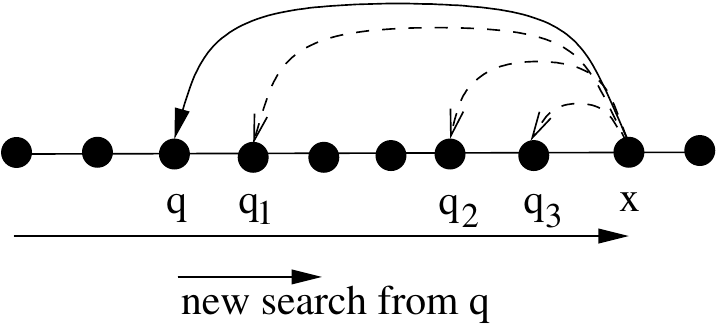}
\caption{Amortized complexity of the forward search. The search starts
  again from $q$. On this instance $l=3$ and $q+3<x.$}
 \label{forwardcomp}
\end{figure}

The important point to notice is that after having identified the
right backward transition from $x$ for $t_{i+1}$ reaching state $q$
(there must be one), the search for $t_{i+2}$ starts from $q< x.$
Moreover, we associate all $l$ transitions $q_k=\delta_k(x,[i,j])$
touched before finding the right one to its ending state which
verifies $q<q_k<x$. Thus $q+\ell<x$. This point is illustrated in Figure
\ref{forwardcomp}. As the search starts again from $q$ and that at
most one forward transition is passed through by text symbol, the
total number of forward and backward transitions touched or passed
through when reading the whole text $t=t_1\ldots t_n$ is thus bounded
by $2n$.
\qed
\end{proof}

\bigskip
\begin{center}
   \rule{10cm}{1pt}.
\end{center}
\bigskip

\begin{proof}[Proof of lemma~\ref{kmp}].
Exactly as in the case of a classical text, we amortize the complexity
of the search over the number of transitions we pass through and the
number of reinitialisations of the search we do if no more failure
transition is available. Each time we pass through a failure
transition, we decrease the state from where we will go on the search
if the state is validated. Thus, there can be at most as many failure
transitions passed through during the whole reading of the text as the
number of forward transitions that has been passed through. Since this
number is at most the size of the text, the total number of
transitions touched is at most $2n$. Then, if after a descent from
failure transition to failure transition no more outgoing transition
exists, we reinitialise the search to state 1. Thus there are at most
$n$ such reinitialisations and the total complexity of transitions and
states touched is bounded by $3n$.
\qed
\end{proof}

\bigskip
\begin{center}
   \rule{10cm}{1pt}.
\end{center}
\bigskip

\begin{proof}[Of Lemma~\ref{kmpbuild}]
Before processing, the pattern we first reduce the range of the keys
from $[n]$ to $[m]$. This is done in deterministic
$O(m\log\log m)$ times by first sorting the keys using the fastest integer sorting
algorithm due to Han~\cite{Ha02}, and then replacing each key by its
rank obtained from the sorting.

We then process the pattern in left-to-right in $m$ steps and at each step $j$
determine the failure and forward transitions outgoing of state $j$. We
use two predecessor data structures that require $O(m)$ words of
space and support insert, delete and query operations (a query
operation returns both the predecessor and the successor) in (worst-case)
time $O(\log\log m)$. As we move forward in the pattern, we insert each
symbol in both predecessor data structures (except for the first
symbol which is only inserted in the first predecessor data
structure). The difference between the two predecessor data structures
is that the first one will only get insertions while the second one
can also get deletions. The first is used to determine
forward transitions while the second one is used to determine
failure transitions.

We now show how we determine the transitions at each step $j$. The
forward transitions connecting state $j$ to state $j+1$ is labeled by
$\mbox{rep}(p,j+1)$. The latter is determined by doing a predecessor
search for $p_{j+1}$ on the first predecessor data structure. This
gives us both the predecessor and successor of $p_{j+1}$ among
$p_1\ldots p_{j}$ which is exactly $\mbox{rep}(p,j+1)$.

The failure transition is determined in the following way. If the
target state of the failure transitions of state $j-1$ is state
$i$. Then we do a predecessor query on the the second predecessor data
structure. If the pair of returned prefixes is precisely
$\mbox{rep}(p,i+1)$, then we can make $i+1$ as a target for state
$j$. Otherwise we take the failure transition of state $j-1$. If that
transitions leads to a state $k$, then we remove the symbols
$p_{j-i}..p_{j-k}$ from the second predecessor data structure.
\qed
\end{proof}

\bigskip
\begin{center}
   \rule{10cm}{1pt}.
\end{center}
\bigskip

\begin{proof}[Of Property~\ref{forward}]
We first build the Morris-Pratt representation in $O(m\log\log m)$
time. We then consider each state $x>0$ corresponding to the $p_1
\ldots p_x$ from left to right and for each such state we expand its
backward transitions. Let us sort all $p_i, 1 \leq i \leq x$ and
consider the resulting order $p_{i_0}=-\infty<p_{i_1}<\ldots
<p_{i_k}<+\infty=p_{i_{k+1}}$. We build one outgoing transition for
each interval $[p_{i_j},p_{i_{j+1}}]$, excepted if
$p_{i_{j+1}}=p_{i_{j}}+1.$ This transition is computed as follows. Let
$q$ be the image state of the failure transition from $x$. We pick a
value $z$ in $[p_{i_j},p_{i_{j+1}}]$ an search for $z$ from $q$. Let
$q'$ be the new state reached. We create a backward transition form $x$
to $q'$ labeled $[p_{i_j},p_{i_{j+1}}]$. After this process we created
at most $m^2$ edges in at most $O(m^2\log\log m)$ time.

We now merge backward transitions from the same state to the same
state that are labeled by consecutive intervals. This required at
most $O(m^2)$ time. The whole algorithm thus requires $O(m^2\log\log
m)$ time.
\qed
\end{proof}

\bigskip
\begin{center}
   \rule{10cm}{1pt}.
\end{center}
\bigskip

\begin{proof}[Of Lemma~\ref{pre-counting lemma-1}]
Let A be an algorithm to search for $p$ in $t$ running in $O(l)$
time. It can be converted in an algorithm to search for $p$ inside all
$s_i(t)$ also running in $O(t)$ since: (a) it suffices to remove all
occurrences overlapping two segments and occurrences in the last few
remaining symbols of $t$ out of a segment; and (b) in $O(l)$ time,
only at most $O(l)$ such occurrences can be reported, so only $O(l)$
occurrences might have to be discarded; and (c) testing if an
occurrence is overlapping two segments can be done in constant $O(l)$
time. The extra work required to remove all overlapping occurrences is
therefore also $O(l)$, and thus $A$ can be converted in an $O(l)$ algorithm to search for $p$
inside all segments $s_i(t)$. This implies that a lower bound for this
last problem is also a lower bound for A.
\end{proof}
\bigskip
\begin{center}
   \rule{10cm}{1pt}.
\end{center}
\bigskip

\begin{proof}[Of Lemma~\ref{pre-counting lemma-2}]
All segments $s_i(t)$ are identically distributed, independently of
each other. Thus the average time for searching for $p$ in any segment
is the same. As the expected time is the sum of the expected time to
search for $p$ in all segments, the sum commutes and the expected time
becomes $\lfloor n/(2m-1) \rfloor$ times the average expected time to
search for $p$ in any segment.
\qed
\end{proof}

\bigskip
\begin{center}
   \rule{10cm}{1pt}.
\end{center}
\bigskip

\begin{proof}[Of Lemma~\ref{counting lemma}]
Let $1 \leq i_1 < i_2 \ldots < i_\ell \leq m$ be the position of the accesses.
For $0 \leq j \leq d,$ we define
 $$B_j=\{ b \mid b \in \{1,2,\ldots ,m\} \mbox{ and } j+b =i_t \mbox{
  for some } 1 \leq t \leq \ell \}\; .$$ Note that $|B_j|\leq \ell$
for $1 \leq j \leq d$.  Also, for any $p \in {\cal P}_m(\ell)$, since
it is canceled by the $\ell$ accesses considering isomorphic orders, for
all shift $j$ there is a mismatch, {\em i.e.} there exists two
positions $k,\ell \in B_j$ such that $p[k]>p[\ell]$ and $t[j+k]<
t[j+\ell]$.  We then show that we can find $J \subset
\{0,1,\ldots,d\},$ $|J| = \left \lceil d /\ell^2 \right \rceil$, such
that $B_{j_1} \cap B_{j_2} = \emptyset$ for $j_1 \neq j_2$ in $J$.

We use a greedy procedure to find $J$. Let $j_1=0$.  Inductively,
suppose that we have found $j_1 \dots j_{k-1}$.  Then $j_k$ is
obtained by finding the smallest $j$ such that $B_j$ is disjoint from
the unions of the previous positions we have already chosen, namely
$B= B_{j_1}\cup B_{j_2}\cup \ldots B_{j_{k_1}}\,.$ We claim that this
procedure allows us to find at least $\left \lceil d /\ell^2 \right
\rceil$ such sets.  We prove in fact that $j_k \leq \ell^2(k-1)$ as
long as $\ell^2(k-1) \leq d$.  Observe that $B$ contains at most
$\ell(s-1)$ positions.  We thus claim that at least one of the sets in
${\cal F} = \{B_0,B_1,\ldots, B_{\ell^2(s-1)} \}$ is disjoint from
$B$.  If not, for each $r$, $0 \leq r \leq \ell^2(s-1)$ there exists a
pair $(b,i_t)$ such that $b \in B_r \cap B$ and $r+b=i_t$ for some $1
\leq t \leq \ell$.  So there must exists at least $\ell^2(s-1)+1$ such
pairs, one for each set $B_r$.  But the total number of such pair is
no more than $|B| \cdot \ell \leq \ell^2(s-1)$, a contradiction.

Now take $J \subset \{0,1,\ldots,d\},$ $|J| = \left \lceil d /\ell^2
\right \rceil$, such that $B_{j_1} \cap B_{j_2} = \emptyset$ for $j_1
\neq j_2$ in $J$.  To prove the lemma, consider a random pattern $p$
from $\mathcal{S}_m$ (the set of permutations of size $m$). Then for all shift $j \in \{0
\dots d\}$, there is a mismatch.  So
\begin{align*}
P(p \in {\cal P}_m(\psi))
& = P(\forall j \in \{1 \dots d\} \text{, there is a mismatch})\\
& \leq P(\forall j \in J \text{, there is a mismatch}).
\end{align*}
Notice that for each $j \in \{1 \dots d\}$, the probability that there is no mismatch with $p$ at shift $j$ is $\frac{1}{|B_j|!}$ which is the probability that the permutation formed by the non-$\star$ symbol is the good one. Since all the sets $B_j$ for $j \in J$ are disjoints, we have
\begin{align*}
P(p \in {\cal P}_m(\ell))
&\leq \prod\limits_{j \in J} P(\text{there is a mismatch at shift }j) \\
&\leq \prod\limits_{j \in J} (1-\frac{1}{|B_j|!})
\leq \left( 1-\frac{1}{\ell!} \right) ^{\left\lceil{\frac{d}{\ell^2}} \right\rceil}
\end{align*}
concluding the proof since $|\mathcal{S}_m| = m!\,$.
\qed
\end{proof}

\bigskip
\begin{center}
   \rule{10cm}{1pt}.
\end{center}
\bigskip

Here we prove that $\ell= \frac{b \log m}{\log \log m}$ with some $b = 1 + o(1)$ verify the following inequalities for $m$ large enough:
$\left\lceil \frac{m-1}{\ell^2} \right\rceil \times \frac{1}{\ell!} \leq 1/10~ (ineq. 1)$
and
$$\frac{195}{200}\ \leq \ 1-\left\lceil \frac{m-1}{\ell^2}\right\rceil \times \frac{1}{\ell!}\ \leq \ \frac{99}{100}\ .$$

\bigskip

Let's recall that the Gamma function of Euler $\Gamma$ is an increasing bijection from $\mathbb{R}_{\geq 2}$ to $\mathbb{R}_{\geq 1}$
verifying that $\Gamma(n+1) = n!$ for all $n \in \mathbb{N}$.

Thus the function $s \mapsto s^2 \Gamma(s+1)$ is an increasing bijection from $\mathbb{R}_{\geq 1}$ to $\mathbb{R}_{\geq 1}$.

For all $m \geq 1$, this allows to define $s \in \mathbb{R}_{\geq 1}$ such that $s^2 \Gamma(s+1) = 50 m$.

Thus $s \rightarrow \infty$ when $m \rightarrow \infty$.

Then we set $b = s \times \frac{\log \log m}{\log m}$.

Taking $\ell= s$, then we have $\ell= \frac{b \log m}{\log \log m}$.

Let us prove that $\ell$ satisfied the desired inequalities.

\bigskip

We have $\ell^2 \times \ell ! = s^2 \Gamma(s+1) = 50 m$.

Thus $\left\lceil \frac{m-1}{\ell^2} \right\rceil \times \frac{1}{\ell!} \leq \frac{m}{50 m} \leq \frac{1}{50}$
and $\left\lceil \frac{m-1}{\ell^2} \right\rceil \times \frac{1}{\ell!} \geq \frac{m-1}{50 m} \geq \frac{1}{100}$
for $m$ large enough.
This proves the desired inequalities since $1/50 \leq 1/40 = 1 - 195/200 \leq 1/10$.

\bigskip

Let us prove now that $b = 1 + o(1)$.

By the Stirling inequality, we have that $F(s) < \Gamma(s+1) < 2 F(s)$ with $F(s) = \left(\frac{s}{e}\right)^s\sqrt{2\pi s}$.

Thus  $\sqrt{2\pi} \exp(G(s)) < s^2 \Gamma(s+1) < 2 \sqrt{2\pi} \exp(G(s))$\\ with $G(s) = -s + \left(s+\frac{5}{2}\right) \log(s)$.

We deduce that $\frac{25m}{\sqrt{2\pi}} < \exp(G(s)) < \frac{50m}{\sqrt{2\pi}}$.

Therefore $\log\left(\frac{25}{\sqrt{2\pi}}\right) + \log m < G(s) < \log\left(\frac{50}{\sqrt{2\pi}}\right) + \log m$.

Recall that $f(m) \sim g(m)$ means that $f(m) = g(m) + o(g(m))$ when $m \rightarrow \infty$.

Thus we have $G(s) \sim \log m$.

It is then enough to prove that $G(s) \sim b \log m$. Indeed this imply $b \sim 1$, i.e., $b = 1 + o(1)$.

But $G(s) = -s + \left(s+\frac{5}{2}\right) \log(s) = s \log s + o(s \log s)$.

Thus $s = o(G(s))$, i.e., $s = o(\log m )$.

Since $s =  \frac{b \log m}{\log \log m}$, this means that $b = o(\log \log m)$

Moreover  $\log s =  \log b + \log \log m - \log \log \log m = \log \log m + o(\log \log m)$.

As $G(s) = -s + \left(s+\frac{5}{2}\right) \log(s)$ we then have:

$G(s) = -\frac{b \log m}{\log \log m} + \left( \frac{b \log m}{\log \log m} +\frac{5}{2}\right)\Big( \log \log m + o(\log \log m) \Big)$.

Thus $G(s) \sim b \log m$, concluding the proof.

\bigskip
\begin{center}
   \rule{10cm}{1pt}.
\end{center}
\bigskip


\end{document}